\documentclass[aps,prl,twocolumn,superscriptaddress,showpacs,final,nofootinbib,floatfix]{revtex4}

\usepackage{amsmath,amsfonts,amsthm,amssymb}
\usepackage{graphicx}
\usepackage{braket}
\usepackage{setspace}
\usepackage{hyperref}
\usepackage{bbm}
\usepackage{color}
\DeclareMathOperator{\tr}{tr}
\newcommand{\ie}{\textit{i.e.}}
\newcommand{\eg}{\textit{e.g.}}
\newtheorem{proposition}{Proposition}
\newtheorem{lemma}{Lemma}
\newtheorem{remark}{Remark}

\begin{document}

\title{Steering maps and their application to dimension-bounded steering}

\author{Tobias Moroder}
\affiliation{Naturwissenschaftlich-Technische Fakult\"at, 
Universit\"at Siegen, Walter-Flex-Str.~3, 57068 Siegen, Germany}

\author{Oleg Gittsovich}
\affiliation{Institute of Atomic and Subatomic Physics, TU Wien, 
Stadionallee 2, 1020 Wien, Austria}
\affiliation{Institute for Theoretical Physics, University of Innsbruck, 
Technikerstr. 25, 6020 Innsbruck, Austria}
\affiliation{Institute for Quantum Optics and Quantum Information, 
Austrian Academy of Sciences, Technikerstr. 21a, 6020 Innsbruck, Austria}

\author{Marcus Huber}
\affiliation{Departament de F\'isica, Universitat Aut\`onoma de Barcelona, 
08193 Bellaterra, Spain}
\affiliation{ICFO-Institut de Ci\`encies Fot\`oniques, 
Mediterranean Technology Park, 08860 Castelldefels (Barcelona), Spain}
\affiliation{Group of Applied Physics, University of Geneva, 1211 Geneva 4, Switzerland}

\author{Roope Uola}
\affiliation{Naturwissenschaftlich-Technische Fakult\"at, 
Universit\"at Siegen, Walter-Flex-Str.~3, 57068 Siegen, Germany}

\author{Otfried G\"uhne}
\affiliation{Naturwissenschaftlich-Technische Fakult\"at, 
Universit\"at Siegen, Walter-Flex-Str.~3, 57068 Siegen, Germany}

\begin{abstract}
The existence of quantum correlations that allow one party to steer
the quantum state of another party is a counterintuitive quantum effect that 
has been described already at the beginning of the past century. 
Steering occurs if entanglement can be proven although the description of 
the measurements on one party is not known, while the other side is characterized. 
We introduce the concept of steering maps that allow to unlock the sophisticated techniques 
developed in regular entanglement detection to be used for certifying 
steerability. As an application we show that this allows to go 
even beyond the canonical steering scenario, enabling a generalized 
dimension-bounded steering where one only assumes the Hilbert space 
dimension on the characterized side, but no description of the measurements.
Surprisingly this does not weaken 
the detection strength of very symmetric scenarios that have recently been carried out in 
experiments. 
\end{abstract}


\maketitle

\textit{Introduction.}---While the term steering was coined already 
in the early days of quantum mechanics \cite{schroedinger35a}, its 
precise treatment only started alongside modern developments in 
quantum information theory \cite{wiseman07a,brunner_review}. The 
possibility to steer the ensemble in a two-party shared state in 
quantum mechanics requires that the two subsystems are 
entangled. To show steering, however, entanglement is not sufficient, 
since there are even some entangled states that are non-steerable. 
In fact, steering can be seen as 
entanglement verification where one relaxes all assumptions about 
the devices used by one of the parties, thus sacrificing the
ability to detect all entangled states.

This fundamental fact is also what motivates one of the recent interests 
into certifying the steerability of quantum states: Any successful steering 
test constitutes an entanglement test that is completely device independent 
for one of the parties and can thus be exploited to design more secure quantum protocols in situations where one of the parties may be untrusted. Apart from 
this it has been observed recently that steering is fundamentally
asymmetric~\cite{bowles14a} and that it is closely connected to joint measurability~\cite{quintino14a,uola14a}. Furthermore, steering is known to 
give an advantage for tasks like subset channel discrimination~\cite{piani14a}.  
Naturally this also spurred the interest in devising strong steering  criteria~\cite{wiseman07a,cavalcanti09a,schneeloch13a,evans13a,horodecki14a,zukowski14a}, to investigate their violation~\cite{marciniak14a} or to develop and to use 
it quantitatively~\cite{skrzypczyk13a,steeringGauss,pusey13a}. It has been shown 
that also bound entangled quantum states exhibit steering~\cite{moroder14a}. Experimentally, steering has been successfully shown in several recent experiments~\cite{bennet12a, wittmann12a, smith12a}, which all demonstrate 
that steering, taking into account also various loopholes, is already reachable 
with today's technology. 

In this manuscript we operationally connect steering with regular entanglement 
verification: We develop a framework that maps the steering certification problem to a regular entanglement detection problem.  More explicitly we construct a matrix from the measurement data that exhibits entanglement if the state is steerable. These steering maps, like we call them, allow us to harness the sophisticated techniques developed in entanglement theory and to go beyond the current state of the art in steering. Contrary to intuition this does not complicate the construction of steering criteria at all. In fact, we can use the resulting entanglement tests to derive non-linear or other improved steering tests that are not straightforward to derive with the standard semidefinite programming (SDP) approach at no additional expense. As an example of the vast possibilities of this framework we introduce a new concept that we call dimension-bounded steering and show that it is accessible with our techniques. In this scenario one removes also all assumptions of the usually trusted side, except that all  
measurements 
operate in the same Hilbert space of dimension $d$. In that, this dimension-bounded steering lies between nonlocality and regular steering. Nonetheless we also show that the robustness to experimental noise of dimension-bounded steering can be comparable or even equal to regular steering certification. This implies that recent loophole-free steering experiments could have also shown loophole-free dimension-bounded steering.

The manuscript is organized as follows: We first define steering and set 
the notation. We continue by demonstrating our approach in a dichotomic 
setting and then discuss our main technique, the steering maps. With this 
we then show that deciding steerability of an ensemble is equivalent to 
a separability problem. In the later part we discuss how our approach 
can be used to derive criteria for the dimension-bounded case. We end 
with an explicit example of this criterion for recent experiments and 
a discussion on its strength.

\textit{Steering.}---In the steering scenario, two parties (Alice and Bob) 
share a quantum state $\rho$. Alice
can choose between $n$ different measurements, each having $m$ 
possible results. Her choice is denoted by $x=1,\dots,n$ for the setting while 
the results are labeled by $a=1,\dots,m$. For Bob we assume that he performs 
full tomography on his reduced state depending on Alice's measurement and result.
So he is able to reconstruct the conditional states $\rho_{a|x}$ and  the 
data of this experiment is summarized by the ensemble 
$\mathcal{E}=\{ \rho_{a|x}\}_{a,x}$ of unnormalized density operators, 
where Alice's probability is $P(a|x)=\tr(\rho_{a|x})$. 

Originally, the question of steering asks whether Alice can convince Bob 
that she can steer the state at Bob's side via her measurements. This 
means that Bob cannot explain the reduced states $\rho_{a|x}$ as coming 
from some probability distribution $p(\lambda)$ of states $\rho_\lambda$, 
where Alice's measurements just give additional information about the
probability. As shown in Ref.~\cite{pusey13a} this can be reformulated
as follows:
An ensemble $\mathcal{E}$ is non-steerable if and only if there exist 
unnormalized density operators $\omega_{i_1\dots i_n}$ with $i_k=1,\dots,m$ 
for each $k=1,\dots,n$ such that
\begin{equation}
\label{eq:def_lhs}
\rho_{a|x} = \sum_{i_1,\dots,i_n} \delta_{i_x,a} \omega_{i_1\dots i_n}
\end{equation}
and steerable otherwise. This is the definition from which we start our 
considerations.

\textit{A dichotomic warm up.}---Let us first discuss the idea via the most 
simplest scenario of Alice having two dichotomic measurements, \ie, $n=m=2$,
in which case we use labels $a=\pm$ to provide easier distinguishable formulae. 
In this scenario the ensemble $\mathcal{E}=\{\rho_{+|1},\rho_{-|1},\rho_{+|2},\rho_{-|2} \}$ 
is called non-steerable if and only if there exists positive semidefinite operators 
$\omega_{ij}$ with $i,j=\pm$ such that
\begin{equation}
\label{eq:cond1}
\begin{aligned}
\rho_{+|1}&=\omega_{++}+\omega_{+-},  &\rho_{+|2}&=\omega_{++}+\omega_{-+},  \\
\rho_{-|1}&=\omega_{-+}+\omega_{--},  &\rho_{-|2}&=\omega_{+-}+\omega_{--},  
\end{aligned}
\end{equation}
holds. Note that these linear equations are not linearly independent, therefore 
$\mathcal{E}$ does not completely determine the unknowns $\omega_{ij}$. Choosing 
for instance an arbitrary $\omega_{++}$ the choices
\begin{equation}
\label{eq:lin_sol}
\begin{aligned}
\omega_{++}&,  &\omega_{+-}&=\rho_{+|1}-\omega_{++},  \\
\omega_{-+}&=\rho_{+|2}-\omega_{++},  &\omega_{--}&=\rho_{\Delta}+\omega_{++},  
\end{aligned}
\end{equation}
with $\rho_{\Delta}=\rho-\rho_{+|1}-\rho_{+|2}$ satisfy the linear constraints, where $\rho$ denotes the reduced density matrix of Bob.

Recall that steering constitutes one-side device-independent entanglement verification, 
because a non-steerable ensemble can always be reproduced by measurements on a 
separable state $\sigma_{AB}$. This works by the using
\begin{equation}
\label{eq:separable_state}
\sigma_{AB} = \sum_{ij} \ket{i,j}_{A}\bra{i,j} \otimes \omega_{ij}, 
\end{equation}
where $\ket{\pm,\pm}_A$ label computational basis states and 
measurements
$M_{\pm|1}= \ket{\pm}\bra{\pm}\otimes \mathbbm{1}$, $M_{\pm|2}=\mathbbm{1}\otimes \ket{\pm}\bra{\pm}$.

Whether we explicitly search for appropriate $\omega_{ij}$ satisfying Eq.~\eqref{eq:cond1}
or for the separable state $\sigma_{AB}$ in Eq.~\eqref{eq:separable_state} one could guess 
there is not much difference. However, looking for a separable state is a task we are 
well familiar with nowadays, due to extensive research in the past two decades on 
separability criteria~\cite{horodecki_review, guehne09a}. But there are two things 
to take into account: Obviously the state $\sigma_{AB}$ is not completely known to 
us. Also, $\sigma_{AB}$ is not just a separable state, because Alice's states are 
very special; such states are called classical-quantum~\cite{piani08a} or to have 
zero ``quantum discord"~\cite{dakic10a, navascuesrat}. Thus if one na\"ively applies 
a separability criterion one looses this required extra structure and the criterion 
will not be very strong. In the following we show how to circumvent these drawbacks.

\textit{Steering maps.}---
In the following, we reformulate the original SDP in an equivalent 
manner by using the duality of semidefinite programs~\cite{vandenberghe96a}. 
This will later allow to treat dimension-bounded steering.
First, to remove the discord zero structure we replace 
the basis states $\ket{i,j}\bra{i,j}$ by other positive semidefinite 
operators $Z_{ij}$ of our choice, so that we get a generic separable 
structure 
\begin{equation}
\label{eq:SIGMA}
\Sigma_{AB} = \sum_{ij} Z_{ij} \otimes \omega_{ij}.
\end{equation}
To get a unit trace for $\Sigma_{AB}$ and to remove the problem that not 
all $\omega_{ij}$ are known one enforces certain linear relations on $Z_{ij}$. 
Using for instance the solution of Eq.~\eqref{eq:lin_sol} in Eq.~\eqref{eq:SIGMA} 
one obtains
\begin{align*}
\Sigma_{AB} =& Z_{+-} \otimes \rho_{+|1} + Z_{-+} \otimes \rho_{+|2} +Z_{--} \otimes \rho_{\Delta} \\
\nonumber
&+ (Z_{++}-Z_{+-}-Z_{-+}+Z_{--}) \otimes \omega_{++}, 
\end{align*}
from which one sees that $\Sigma_{AB}$ is completely determined if the 
last term vanishes, \ie, $Z_{++}=Z_{+-}+Z_{-+}-Z_{--}$. With this identity 
the normalization of $\tr(\Sigma_{AB})=1$ is then equal to
$\tr(Z_{+-})\tr(\rho_{+|1})+\tr(Z_{-+})\tr(\rho_{+|2})+\tr(Z_{--})\tr(\rho_{\Delta}) = 1.$ 
This is exactly what we were looking for and we get the following sufficient
criterion for steerability:
\textit{For any non-steerable ensemble $\mathcal{E}$ and any choice of 
positive semidefinite operators $ Z_{ij}$, which satisfy the two just 
mentioned extra relations, the operator 
\begin{equation}
\label{eq:Z_sol}
\Sigma_{AB}=Z_{+-} \otimes \rho_{+|1} + Z_{-+} \otimes \rho_{+|2} +Z_{--} \otimes \rho_{\Delta} 
\end{equation}
is a separable quantum state.}

If for a given set of $Z_{ij}$ the state $\Sigma_{AB}$ is not separable, 
\ie, entangled or no quantum state at all, then operators $\omega_{ij}$ with 
the properties from Eqs.~(\ref{eq:cond1}, \ref{eq:lin_sol}) do not exist and
the underlying ensemble is steerable. In order to check this we can employ 
any separability criterion, \eg, partial transposition~\cite{peres96a}, positive 
maps~\cite{horodecki96b}, entanglement witness~\cite{horodecki96b,terhal00a}, 
computable cross norm or realignment~\cite{rudolph02a,chen03a}, covariance 
matrices \cite{PhysRevLett.99.130504}, to name only a few. The whole power of 
this is unlocked by the mapping $\ket{i,j}\bra{i,j} \mapsto Z_{ij}$, which we 
refer to as \emph{steering map} from now on. 

In the most general steering case 
we know that a non-steerable ensemble can always be obtained by measuring 
the separable state
%\begin{equation}
$
\sigma_{AB} = 
\sum_{i_1\dots i_n} \ket{i_1,\dots,i_n}_A\bra{i_1,\dots, i_n} \otimes \omega_{i_1\dots i_n},
$
%\end{equation}
with appropriate measurements that only act non-trivially on the respective 
subsystem for Alice. Each computational basis state is now mapped to a new 
positive semidefinite operator $Z_{i_1\dots i_n}$ to obtain
\begin{equation}
\label{eq:gen_sigma}
\Sigma_{AB} = \sum Z_{i_1\dots i_n} \otimes \omega_{i_1 \dots i_n}.
\end{equation}
This operator is uniquely determined by the given ensemble 
$\mathcal{E}$ if and only if the chosen operators $Z_{i_1 \dots i_n}$ 
satisfy  
\begin{align}
\nonumber
Z_{i_1 i_2 \dots i_n} = & 
Z_{i_1 j_2 \dots j_n} + Z_{j_1 i_2 j_3 \dots j_n} + \dots  + Z_{j_1 j_2 \dots i_n}- \\ 
\label{eq:condZ}
& - (n-1) Z_{j_1 j_2 \dots j_n}
\end{align} 
for all possible choices of $i_1,\dots, i_n$ and $j_1,\dots j_n$. With this we 
are ready to state our first main result, which says that the developed criterion 
via steering maps is also sufficient. The proof 
is given in the appendix.

\begin{proposition}\label{prop:steering_map}
For any non-steerable ensemble $\mathcal{E}$ and any set of positive 
semidefinite operators $\mathcal{Z}=\{Z_{i_1\dots i_n}\}_{i_1\dots i_n}$ 
fulfilling (\ref{eq:condZ}) the operator given by Eq.~\eqref{eq:gen_sigma} 
has a separable structure. 

For any steerable ensemble $\mathcal{E}$ there exists a set of operators 
$\mathcal{Z}$ which uniquely determines $\Sigma_{AB}$ and satisfies 
$\tr(\Sigma_{AB})=1$, but where non-separability of $\Sigma_{AB}$ 
is detected by the swap entanglement witness. Here, the 
swap entanglement witness is the flip operator $V=\sum_{ij}\ket{ij}\bra{ji}$ 
where $Tr(\rho V) < 0$ signals entanglement.
\end{proposition}

Let us remark that the steering map criterion is strictly stronger than a 
single steering inequality, which is similarly characterized by $\mathcal{Z}$, 
but where one only checks the swap entanglement witness. Moreover, the proposition 
also applies to steering scenarios where Bob measures a few observables rather 
than a tomographic complete set; in this case  non-separability of $\Sigma_{AB}$ 
must be verified via this partial information only. Note that since steering 
is closely related to joint measurability, Prop.~\ref{prop:steering_map} can 
directly be employed also for this task, and we are using a result from this 
field~\cite{carmeli12a} to deduce a collection of $\mathcal{Z}$ for the 
case $n=2, m=d$, cf. appendix.

\textit{Dimension-bounded steering.}---Next let us turn to the dimension-bounded 
steering case. Contrary to the standard steering setup, where it is essential that 
the measured observables on the characterized side are fully known, these criteria 
require only that Bob's measurements act on a fixed finite dimensional Hilbert space. 

To be precise, we assume that Bob can choose between $n_B$ different 
settings $y$ each yielding one of $m_B$ possible outcomes $b$. Each 
measurement is described by a POVM, \ie, a set of operators $\{M_{b|y}\}_b$
which satisfies positivity $M_{b|y}\geq0$ and normalization 
$\sum_b M_{b|y}=\mathbbm{1}$. As the sole restriction we have to assume
that they all act on the same Hilbert space with at most dimension $d_B$. 
Thus if Bob observes different distributions, $P(b|y,i)$, maybe conditioned
onto a separate event $i$ like a measurement result by Alice, then there must
exist a collection of different density operators $\{ \rho_i \}_i$ and a 
single set of appropriate POVMs, both on an $d_B$-dimensional Hilbert space, 
which reproduce the data, $P(b|y,i)=\tr(M_{b|y} \rho_i)$. \footnote{Note that 
we do not ``convexify'' the set of possible distributions, \ie, we are not
assuming the more general form 
$P(b|y,i)=\sum_\lambda P(\lambda) \tr(M_{b|y;\lambda} \rho_{i;\lambda})$ 
with  $d_B$-dimensional quantum states and measurements. First, we consider this 
largely unmotivated for experiments, second, it would considerably weaken 
the detection strengths of the criteria, and third, since it effectively 
corresponds to the case of many different $d_B$-dimensional systems it 
is a strange dimension restriction, except if one distinguishes classical
and quantum dimensions~\cite{brunner08a}.} To complete the description 
of the problem we 
assume that $n_A$, $m_A$ are the subsystem-labeled specifications 
for Alice, who is the fully uncharacterized side, and refer to 
it as a $d_B$-dimension-bounded steering scenario with
parameters $n_A,m_A,n_B,m_B$.

In order to derive steering criteria for this scenario we employ 
a fixed steering map to transform the problem to a standard 
separability question according to Prop.~\ref{prop:steering_map}.
Afterwards we use the entanglement detection techniques of 
Ref.~\cite{moroder12a} which require only a dimension constraint.

The criteria that we derive work best for Bob having dichotomic 
measurements  $n_B=2$. Before we give the main recipe we like to 
explain the ideas: As shown in the previous section we know that 
any steerable ensemble $\mathcal{E}$ can be detected by an 
appropriate collection $\mathcal{Z}$ such that 
%\begin{equation}
$
\Sigma_{AB} = \sum_{i_1\dots i_n} Z_{i_1\dots i_n} \otimes \omega^{\rm spec}_{i_1\dots i_n}
$
%\end{equation}
is not a separable state. Here, $\omega^{\rm spec}_{i_1\dots i_n}$ 
should express that the $\omega_{i_1\dots i_n}$, when using a $\mathcal{Z}$ 
satisfying Eq.~\eqref{eq:condZ}, is given by a special solution of the linear
relations given by Eq.~\eqref{eq:def_lhs}, \eg, like in Eq.~\eqref{eq:Z_sol}. 
To show that $\Sigma_{AB}$ is not separable we can employ the CCNR 
criterion~\cite{rudolph02a,chen03a}. This criterion states that the correlation matrix $[C(\rho_{AB})]_{kl} = \tr( G_k^A \otimes G_l^B \rho_{AB})$ of any separable state $\rho_{AB}^{\rm sep}$ satisfies $\| C(\rho_{AB}^{\rm sep}) \|_1 \leq 1$. Here the appearing norm is the trace norm $\| C \|_1 = \sum_i s_i(C)$ given by the sum of the singular values $s_i(C)$, while the sets $\{ G_i \}_i$ are orthonormal Hermitian operators (not necessarily forming a basis) for the respective local side. Thus whenever $\| C(\Sigma_{AB}) \|_1 > 1$ the data $\mathcal{E}$ shows steering. Note, since $\| \cdot \|_1$ is unitarily equivalent, only the corresponding spanned local operator spaces matter.

However, one cannot directly evaluate this for the dimension-bounded scenario, 
because Bob can neither reconstruct $\rho_{a|x}$ nor compute values 
$\tr(G_k^B \rho_{a|x})$ because he lacks the precise description of
his measurements $M_{b|y}$. Still, we can build a matrix which looks 
similar to the correlation matrix and for which the dichotomic choice 
of Bob's measurements becomes important. For each dichotomic measurement
consider the operators given by the difference of the two POVM elements 
$B_y=M_{+|y}-M_{-|y}$ for $y=1,\dots, n_B$ and $B_0 =\mathbbm{1}$. Then, 
define the matrix $[D(\Sigma_{AB})]_{ky}$ 
with entries
\begin{align}
\label{eq:data_matrix}
\!\!\tr(G_k^A\otimes\! B_y \Sigma_{AB})\!=
\!\! \sum_{i_1\dots i_n} \!\tr(G_k^A Z_{i_1\dots i_n}) \tr(B_y \omega^{\rm spec}_{i_1\dots i_n}\!).
\end{align}
For convenience we assume that we only pick $n_B+1$ different 
operators $G_k^A$, such that $D$ is a square matrix with a determinant. 
We call this matrix the \emph{data matrix} $D$ to further express that 
$D$ is determined by the observed data $P(a,b|x,y)$ once
having selected $\mathcal{Z}$ and$\{ G_k^A \}_k$. 

From the data matrix $D$ we obtain a correlation matrix $C=DT$ if $T$ 
describes a linear transformation that maps $\{ B_y \}_y$ into an 
orthonormal set $\{ G_l^B=\sum_y T_{yl} B_l \}_l$. Though having only 
the limited information about $n_B$ being dichotomic measurements 
on a $d_B$-dimensional Hilbert space, this transformation $T$ satisfies \cite{moroder12a}
\begin{equation}
\label{eq:detT}
| \det(T) | \geq d_B^{-\frac{n_B+1}{2}}.
\end{equation}
To be precise, this only holds if $\{ B_y \}_y$ is linearly independent, but which can inferred directly from a data matrix with $|\det(D)| \not = 0$. Via this one can 
then lower bound the trace-norm of $C$ by
\begin{align}
\nonumber
\| C \|_1 & = \sum s_i(C) \geq (n_B+1) | \det(C) |^{\frac{1}{n_B+1}} \nonumber\\&= (n_B+1) \left( |\det(D)| |\det(T)| \right)^{\frac{1}{n_B+1}} \nonumber\\
\label{eq:motiv1}
&\geq \frac{n_B+1}{\sqrt{d_B}} |\det(D)|^{\frac{1}{n_B+1}}
\end{align}
using the inequality of the arithmetic and geometric means in the first 
step, the determinant rule, and finally Eq.~\eqref{eq:detT}. If this lower
bound is strictly above $1$, we certify that $\Sigma_{AB}$ is not separable 
and thus steerability of the underlying state. This is effectively the second 
condition of the following proposition; the other statement employs a 
slightly better bounding technique. 

\begin{proposition}\label{prop:dbsteering}
Consider a $d_B$-dimension-bounded steering scenario with parameters 
$n_A,m_A,n_B$ and $m_B=2$. {From} the observed data build up the data matrix
\begin{equation}
D_{ky} = \sum_{i_1\dots i_n} \tr(G^A_k Z_{i_1\dots i_n}) \tr(B_y \omega_{i_1\dots i_n}^{\rm spec})
\end{equation}
using $B_0=\mathbbm{1}$ and $B_y =M_{+|y}-M_{-|y}$ for $y=1,\dots,n_B$, any 
set of steering operators $\mathcal{Z}$ with $n_A, m_A$, and any choice of 
$n_B+1$ orthonormal operators $G_k^A$. 

Let $d_A$ be the dimension of the chosen $\mathcal{Z}$. If the observed data 
are non-steerable then the determinant of $D$ satisfies 
\begin{equation}
\label{eq:bound1}
|\det(D)| \leq \frac{1}{\sqrt{d_A}}\left( \frac{\sqrt{d_Ad_B}-1}{n_B\sqrt{d_A}}\right)^{n_B}
%\frac{1}{\sqrt{d_A d_B}} + n_B\left( \sqrt{d_A} d_B^{-\frac{n_B}{2}} | \det (D)|\right)^{\frac{1}{n_B}} \leq 1
\end{equation} 
if $n_B > \sqrt{d_A d_B} -1$ and $\mathbbm{1} \in \textrm{span}(\{ G_i^A\})$. If 
this is not the case, non-steerable data give
\begin{equation}
\label{eq:bound2}
|\det(D)| \leq \left( \frac{\sqrt{d_B}}{n_B+1}\right)^{n_B+1}.
%\frac{n_B+1}{\sqrt{d_B}} | \det(D) |^{\frac{1}{n_B+1}} \leq 1
\end{equation}
\end{proposition}

\textit{Application to experiments}.---In this part we give an explicit example of 
Prop.~\ref{prop:dbsteering} to demonstrate its application and also to compare its 
strength. We pick the scenario that has been implemented in the loophole-free steering 
experiment performed in Vienna~\cite{wittmann12a}. 
We follow the procedure outlined in our manuscript to arrive at the data matrix 
(for details see the Appendix):
\begin{equation}
\nonumber
\small
\frac{1}{\sqrt{2}}\! \left[\! \begin{array}{cccc}
1 & \braket{B_1} & \braket{B_2} & \braket{B_3} \\
\braket{A_1}/\sqrt{3} & \braket{A_1 B_1}/\sqrt{3} & \braket{A_1 B_2}/\sqrt{3} & \braket{A_1 B_3}/\sqrt{3} \\
\braket{A_2}/\sqrt{3} & \braket{A_2 B_1}/\sqrt{3} & \braket{A_2 B_2}/\sqrt{3} & \braket{A_2 B_3}/\sqrt{3} \\
\braket{A_3}/\sqrt{3} & \braket{A_3 B_1}/\sqrt{3} & \braket{A_3 B_2}/\sqrt{3} & \braket{A_3 B_3}/\sqrt{3} \end{array} \!\right]. 
\end{equation}
Because $n_B=3 > \sqrt{d_Ad_B} - 1 = 1$ and since the full operator basis 
for $A$ includes the identity we can use the bound given by Eq.~\eqref{eq:bound1}. Thus if
\begin{equation}
\label{eq:CRIT1}
| \det(D) | > \frac{1}{108},
\end{equation}
then the observed data show steering under the sole assumption that Bob's 
measurements act onto a qubit. 

If one evaluates this criterion for a noisy maximally entangled 
state $p\ket{\psi^-}\bra{\psi^-}+(1-p)\mathbbm{1}/4$, measuring
along the three spin directions $\sigma_1,\sigma_2,\sigma_3$, one 
verifies steering if $p > 1/\sqrt{3}$. This is  surprising, 
because the visibility to show standard steering, \ie, requiring 
the knowledge that Bob perfectly measures $\sigma_1,\sigma_2,\sigma_3$, 
is exactly the same. Thus, we learn that for this symmetric case, 
the only crucial knowledge of the measurements is that they act 
onto a qubit, but no further characterization is needed. In the 
appendix we discuss this scenario also under experimentally 
realistic conditions showing that todays technology indeed 
allows (or has already allowed) a loophole-free 
dimension-bounded steering experiment.

\textit{Conclusion.}--- We have introduced a framework that 
allows to map the steering problem to a standard separability 
problem.  This opened the possibility to exploit the sophisticated 
tools available in entanglement detection, thereby creating strong 
steering criteria. We showed dimension-bounded steering, as one
particularly further promising application. Considering 
that many quantum protocols require also a certain level 
of trust we believe that this dimension-bounded scenario 
is of high relevance for scenarios where at least one of 
the parties has some degree of confidence of his or her 
local device. We have shown that this ``nearly'' device 
independent scenario is a lot stronger than the still not
attainable full device-independent scenario. It will help 
to make quantum key distribution more robust~\cite{gittsovich12a,woodhead13a} 
and to unify frameworks of resource theories that exist for 
nonlocality \cite{julio} and steering \cite{rodrigo} to
approach a resource theory of partially device independent
entanglement certification.

\begin{acknowledgments}
We would like to thank B.~Wittmann for stimulating discussions. 
This work has been supported by the EU (Marie Curie CIG 293993/ENFOQI, 
STREP ``RAQUEL'' and Consolidator Grant 683107/TempoQ),
the BMBF (Chist-Era Project QUASAR), 
the FQXi Fund (Silicon Valley Community Foundation), the DFG, 
the Austrian Science Fund (FWF), the Spanish ministry of economy 
through the Juan de la Cierva fellowship (JCI 2012-14155), the 
Marie Curie Actions (Erwin-Schr\"odinger-Stipendium J3312-N27), 
and by the Finnish Cultural Foundation. MH furthermore acknowledges funding through the AMBIZIONE grant PZ00P2\_161351 from the Swiss National Science Foundation (SNF).
\end{acknowledgments}

\section{Appendix}

\subsection{Proof of Eq.~(8)}

Let us summarize the statement in the following proposition:

\begin{proposition}\label{prop:Z}
The set $\mathcal{Z}=\{Z_{i_1 \dots i_n}\}_{i_1 \dots i_n}$ 
uniquely determines $\Sigma_{AB}$ if and only if Eq.~(8) in the main text 
holds for any choices of $i_1,\dots, i_n$ and $j_1,\dots j_n$. 
\end{proposition}

Before we prove this proposition let us note a technical lemma,
which will be useful in the following. It describes the most 
general solution of $\omega_{i_1\dots i_n}$ which satisfy the 
relations demanded for a local hidden state model.

\begin{lemma}\label{lemma:w}
Any collection of hidden states $\omega_{i_1\dots i_n}$ which 
satisfies the set of linear equations given by Eq.~(1) in the main text 
for $\mathcal{E}$ can be written as $\omega=\omega^{\rm spec} + \omega^{\rm homo}$. 
A special solution $\omega^{\rm spec}$ is given by 
$\omega^{\rm spec}_{i_1\dots i_n}=0$ for all indices $i_1,\dots,i_n$ except
\begin{equation}
\label{eq:spec1}
\omega^{\rm spec}_{a m \dots m}=\rho_{a|1},\omega^{\rm spec}_{m a m \dots m} = \rho_{a|2},  \dots \: \omega^{\rm spec}_{m \dots ma}\! =\! \rho_{a|n},
\end{equation}
for $a < m$ and 
\begin{equation}
\label{eq:spec2}
\omega^{\rm spec}_{m \dots m} = \sum_x \rho_{m|x} - (n-1) \rho.
\end{equation}
The general solution of the corresponding homogeneous system is given by
\begin{equation}
\label{eq:sol_homo}
\omega^{\rm homo}_{i_1 \dots i_n} = \sum_{{\bf k}} v^{({\bf k})}_{i_1 \dots i_n} X_{{\bf k}}
\end{equation} 
using arbitrary Hermitian operators $X_{{\bf k}}$. Here ${{\bf k}}=k_1\dots k_n$ 
is an $n$-length index similar to the subscripts of $\omega$, where only the 
distinct possibilities with at least two $k_i<m$ are considered. For a 
fixed ${{\bf k}}$ the vector $v^{({\bf k})}$ is given by
\begin{eqnarray}
\label{eq:v} 
v^{({\bf k})}_{i_1\dots i_n}=&\delta_{i_1\dots i_n,k_1\dots k_n}- \delta_{i_1\dots i_n,k_1m\dots m}
- \ldots \nonumber\\& - \delta_{i_1\dots i_n,m \dots mk_n } + (n-1) \delta_{i_1\dots i_n,m\dots m }.
\end{eqnarray}
\end{lemma}

\begin{proof}
Note that Eq.~(1) in the main text is a standard set of linear equations, 
except that we have Hermitian operators rather than scalar variables. 
Therefore all the basic linear algebra results apply. 

In total we have $m^n$ unknowns but only $n(m-1)+1$ linear independent 
relations recalling once more that $\sum_a \rho_{a|x}=\rho$ is independent 
of the setting. Hence the general solution can be written as a combination 
of a special solution and the general solution of the homogeneous system
$\sum \delta_{i_x,a} \omega_{i_1\dots i_n} = 0$. 

That $\omega^{\rm spec}$ as given in the Lemma is a special solution 
can be checked straightforwardly. For the general solution of the 
homogeneous system $\omega^{\rm homo}$ note that via the Ansatz 
of Eq.~(\ref{eq:sol_homo}) this breaks down to the relation 
\begin{equation}
\label{eq:homo_vec}
\sum_{i_1,\dots,i_n} \delta_{i_x,a} v^{({\bf k})}_{i_1\dots i_n}=0.
\end{equation}
The dimension of this linear subspace is $m^n - [n(m-1)+1]$, which 
is precisely the number of the considered ${\bf k}$'s. Now 
first note that the given $\{ v^{({\bf k})} \}_{\bf k}$ are 
linearly independent, since vector $v^{({\bf k})}$ is the 
only vector which has a non-zero entry at the position $i_1\dots i_n=k_1\dots k_n$.
Thus we are left to show that they indeed solve Eq.~(\ref{eq:homo_vec}). 
For the $x=1$ and $a<m$ this follows for instance by
\begin{equation}
\sum_{i_2 \dots i_n} v^{({\bf k})}_{a i_2 \dots i_n} =
\underbrace{+1}_{a k_2 \dots k_n} \underbrace{-1}_{a m \dots m}  = 0 
\end{equation}
if $k_1=a$, otherwise it holds trivially. 
The same arguments holds if one picks a different index $i_x$. 
At last we still need to check the relation corresponding 
to reduced state, which is given by
\begin{equation}
\sum_{i_1 \dots i_n} v^{({\bf k})}_{i_1  \dots i_n} =
\underbrace{+1}_{k_1 k_2 \dots k_n} \underbrace{-n}_{\{k_1 m \dots m, \dots, m \dots k_n\}}   \underbrace{n-1}_{m\dots m}= 0. 
\end{equation}
which finishes the proof.
\end{proof}

\begin{proof}[Proof of Prop.~\ref{prop:Z}]
Using the general solution $\omega^{\rm sol}$ as given the Lemma~\ref{lemma:w}
in the operator $\Sigma_{AB}$ one sees 
that 
\begin{eqnarray}
\Sigma_{AB} &=& 
\sum_{i_1 \dots i_n } Z_{i_1\dots i_n} \otimes 
\omega^{\rm spec}_{i_1\dots i_n} \nonumber\\ 
&&+\sum_{{\bf k}} \left( \sum_{i_1\dots i_n} v_{i_1\dots i_n}^{({\bf k})}Z_{i_1\dots i_n} \right) \otimes X_{\bf k}
\end{eqnarray} 
is uniquely determined by the given ensemble $\mathcal{E}$ if and only if 
\begin{equation}
\label{eq:constZ}
\sum_{i_1\dots i_n} v_{i_1\dots i_n}^{({\bf k})}Z_{i_1\dots i_n} = 0
\end{equation}
holds for all possibilities ${\bf k}$. Using the explicit form
of the vectors $v^{({\bf k})}$ as given in Eq.~\eqref{eq:v} 
these constraints can be re-written as
\begin{align}
\nonumber
Z_{k_1\dots k_n} =& Z_{k_1 m \dots m} + Z_{m k_2 \dots m} + \ldots + Z_{m \dots k_n} \\
\label{eq:completeZ}&- (n-1) Z_{m \dots m}
\end{align}
for all admissible $k_1\dots k_n$ with at least two $k_i<m$. However, 
this condition also holds also for each $k_1\dots k_n$ without this
restriction, because then the vectors $v^{({\bf k})}$ in Eq.~(\ref{eq:v})
vanish. Thus we have proven Eq.~(8) in the main text for all $i_1\dots i_n$, 
but only for the special index set $j_1 \dots j_n=m\dots m$. Still, these
conditions already imply the general (more symmetric looking) relation, 
using an arbitrary $j_1\dots j_n$. This can be inferred more easily 
directly from the problem formulation by relabeling the individual 
outcomes of the conditional states.
\end{proof}

\subsection{Proof of Proposition~1}

We prove this in two parts; the first only considers the statement 
without the extra condition $\tr(\Sigma_{AB})=1$, but which is 
discussed in the second part then. 

As mentioned in the main text, the proof rests on the duality properties of semidefinite programs. In fact, the first part of the proof can be considered as a special interpretation of the dual program of the original semidefinite program. Since the dual might be of independent interest, we compactly summarizes it in Remark~\ref{remark:direct_dualSDP}.

\begin{proof}[Proof, Part $1$]
The idea of the proof is to employ the duality statements given 
by respective semidefinite programs. Recall that the problems 
$\inf_{x \in \mathbbm{R}^n} \{c^Tx|F_0+\sum_i x_i F_i \geq 0\}$ 
and $\sup_{Z\geq 0 } \{ -\tr(ZF_0)| \tr(ZF_i)=c_i \forall i\}$, 
called primal and dual 
semidefinite programs, are connected by a couple of important 
relations. The most relevant is strong duality, which states 
that both optimal values are equal. This holds for instance 
under the Slater regularity condition that either problem has 
a strictly feasible point, \ie, either an $x$ such that 
$F_0 + \sum_i x_i F_i > 0$ or a $Z>0$ satisfying $\tr(ZF_i)=c_i$~[32].
The proof 
goes along the following lines: We parse the original 
steering problem into the form of the primal semidefinite 
program, then we invoke its dual, show strong duality such 
that we can ensure that it gives the same solution, and 
finally we interpret this dual program as a the swap witness on $\Sigma_{AB}$.

To start let us write the original problem into the form of
a primal semidefinite program, which is given by
\begin{eqnarray}
\label{eq:primal}
\inf & &0 \\
\nonumber
\textrm{s.t.} & & \omega_{i_1\dots i_n}^{\rm spec} + \sum_{{\bf k}} v^{({\bf k})}_{i_1\dots i_n} X_{{\bf k}} \geq 0 \:\: \forall i_1\dots i_n.
\end{eqnarray}
This can be transformed to the standard form if one 
uses, i) a Hermitian operator basis $\{ S_r \}$ to transform
the matrix-valued variables $X_{\bf k}$ into 
$X_{\bf k} = \sum_{r} x_{{\bf k},r} S_r$ to scalar-valued 
variables $x_{{\bf k},r}$, and ii) that several positivity 
constraints are equivalent to a single positivity constraint 
of a corresponding block matrix. We emphasize that 
Eq.~\eqref{eq:primal} is a special primal problem 
called feasibility problem, since we effectively do 
not optimize anything. By convention, if the constraint
cannot be fulfilled then the infimum is $+\infty$.

Working out the dual gives 
\begin{eqnarray}
\label{eq:dual}
\sup & & - \sum_{i_1 \dots i_n} \tr(Z_{i_1\dots i_n} \omega_{i_1\dots i_n}^{\rm spec}) \\
\nonumber
\textrm{s.t.} & & Z_{i_1\dots i_n} \geq 0\:\: \forall i_1\dots i_n, \\
\nonumber
&& \sum_{i} v^{({\bf k})}_{i_1\dots i_n} Z_{i_1\dots i_n} = 0\:\: \forall {\bf k}.
\end{eqnarray}
If one has used the standard form for the previous problem, 
one simply reverses here the points i) and ii); the block-structure 
can be removed directly, while the linear relations in the last 
line of Eq.~\eqref{eq:dual} appear since one has respective 
linear relations for all Hermitian operator basis elements.  

This dual has a strictly feasible point 
$Z_{i_1\dots i_n}=\mathbbm{1} > 0$, noting 
$\sum_{i} v^{({\bf k})}_{i_1\dots i_n}=0$ was already
proven in Lemma~\ref{lemma:w}. Therefore we have strong 
duality, and consequently the statement that, whenever 
the primal problem is infeasible ($\mathcal{E}$ steerable) 
then there exists a sequence of appropriate $Z_{i_1\dots i_n}$ 
such that $C=\sum_{i} \tr(Z_{i_1\dots i_n} \omega_{i_1\dots i_n}^{\rm spec})$ 
will tend to $-\infty$, saying that Eq.~\eqref{eq:dual} 
is unbounded. We summarize this more direct dual SDP in Remark~\ref{remark:direct_dualSDP}.

Now let us interpret this as the detection statement of
the proposition. That we labeled the dual variables by
$Z_{i_1\dots i_n}$ as also used in $\Sigma_{AB}$ is no 
coincidence. Effectively the solutions $Z_{i_1\dots i_n}$
of the dual program will be the ones used in the operator 
$\Sigma_{AB}$ that shows steering. Note that the variables 
of the dual program already satisfy positivity $Z_{i_1\dots i_n}\geq 0$ 
and the linear relations in Eq.~\eqref{eq:dual} uniquely 
determine $\Sigma_{AB} = \sum_{i} Z_{i_1\dots i_n} \otimes \omega_{i_1\dots i_n}^{\rm spec}$,
as already shown in the proof of 
Prop.~\ref{prop:Z}. Finally, note here the formal 
operator connection between $\Sigma_{AB}$ and the 
objective function $C$. Using the swap operator 
$V$, \ie, $\tr(V A \otimes B)=\tr(AB)$, one directly
sees that the swap operator evaluated on $\Sigma_{AB}$
gives the objective value $\tr(V\Sigma_{AB})=C$. Since the 
swap operator $V$ is an entanglement witness a negative $\tr(V\Sigma_{AB})=C<0$
signals that the optimal 
$\Sigma_{AB}$ has not a separable structure. This finishes the first part 
of the proof.\end{proof}

\begin{remark}\label{remark:direct_dualSDP}
The dual problem to the feasibility problem for the collection of positive semidefinite operators satisfying the relations given by Eq.~(1) reads as
\begin{eqnarray}
\sup && - \sum_{i_1 \dots i_n} \tr(Z_{i_1\dots i_n} \omega_{i_1\dots i_n}) \\
\nonumber
\textrm{\emph{s.t.}} & & Z_{i_1\dots i_n} \geq 0\:\:\:\: \forall i_1\dots i_n, \\
\nonumber
& &Z_{i_1 i_2 \dots i_n} = Z_{i_1 j_2 \dots j_n} + Z_{j_1 i_2 j_3 \dots j_n} + \dots  + Z_{j_1 j_2 \dots i_n}\\ 
\nonumber 
& & \:\:\:\:\:\:\:\:\:\:\:\:\:\:\:\:\:\:\:\:\:\:- (n-1) Z_{j_1 j_2 \dots j_n}\;\; \forall i_1,\dots j_n.
\end{eqnarray}
Via the linear equations for $Z_{i_1\dots i_n}$ and by Eq.~(1) one can evaluate the objective $C=\sum_{i_1 \dots i_n} \tr(Z_{i_1\dots i_n} \omega_{i_1\dots i_n})$. For instance, if one picks fixed indices $j_1,\dots,j_n$ one arrives at
\begin{eqnarray*}
C\!\!&=&\!\!\! \sum_{i_1\dots i_n}\! \tr(Z_{i_1 j_2 \dots j_n} \omega_{i_1\dots i_n}) \!+ \dots + \!\!\sum_{i_1\dots i_n} \!\!\tr(Z_{j_1 \dots i_n} \omega_{i_1\dots i_n})\\ && - (n-1) \sum_{i_1\dots i_n} \tr(Z_{j_1 j_2 \dots j_n} \omega_{i_1\dots i_n}) \\
 &=&\!\! \sum_{i_1} \tr\big[Z_{i_1 j_2 \dots j_n} (\sum_{i_2\dots i_n}\omega_{i_1\dots i_n}) \big] + \dots \\ && +\sum_{i_n} \tr[Z_{j_1 \dots i_n} (\sum_{i_2\dots i_n}\omega_{i_1\dots i_n})]\\ && - (n-1)  \tr[Z_{j_1 j_2 \dots j_n} (\sum_{i_1\dots i_n}\omega_{i_1\dots i_n})] \\
&=& \!\! \sum_{i_1} \tr(Z_{i_1 j_2 \dots j_n} \rho_{i_1|1}) + \dots + \sum_{i_n} \tr(Z_{j_1 \dots i_n} \rho_{i_n|n})\\ 
&& -(n-1)\tr(Z_{j_1 j_2 \dots j_n} \rho). 
\end{eqnarray*}
Note that any other choice gives the same value; this is expressed by  $C=\sum_{i_1 \dots i_n} \tr(Z_{i_1\dots i_n} \omega_{i_1\dots i_n}^{\rm spec})$.
\end{remark}

\begin{proof}[Proof, Part $2$]
It is left to show that we can also find a solution $\mathcal{Z}$ 
which satisfies $\tr(\Sigma_{AB})=1$, since such a condition does 
not appear in Eq.~\eqref{eq:dual}. Note that since the value of 
an objective function of any steerable ensemble will tend to $-\infty$, 
there are for sure parameters $\mathcal{Z}$ such that $C < 0$. Suppose 
that for these $Z_{i_1\dots i_n}$, the operator $\Sigma_{AB}$ is not
normalized. If $\tr(\Sigma_{AB})>0$, then one can directly used a 
rescaled version $Z_{i_1\dots i_n}/\tr(\Sigma_{AB})$, now also 
satisfying the trace condition, but still detecting the state. 
Note that this trick fails if $\tr(\Sigma_{AB})\leq 0$, either 
due to a division by zero, or due to $Z_{i_1\dots i_n}$ being not 
positive semidefinite anymore. Thus we are left to prove that
$\tr(\Sigma_{AB}) > 0$.

To verify $\tr(\Sigma_{AB})\geq 0$ we employ that $C\geq 0$ holds for
any non-steerable ensemble. {From} the given ensemble $\mathcal{E}$ such 
a non-steerable ensemble is for instance
$\mathcal{\tilde E}= \{ \tilde \rho_{a|x}= \tr(\rho_{a|x}) \mathbbm{1}/d \}$, 
having a special solution 
$\tilde \omega_{i_1\dots i_n} = \tr(\omega^{\rm spec}_{i_1 \dots i_n})\mathbbm{1}/d$ 
as can be checked by Eqs.~(\ref{eq:spec1}, \ref{eq:spec2}). Thus evaluating the 
objective function of this non-steerable ensemble and the chosen 
selection $\mathcal{Z}$ one finds
\begin{align}
\nonumber
\sum_{i_1\dots i_n}& \tr(Z_{i_1\dots i_n} \tilde \omega_{i_1\dots i_n}^{\rm spec})  \\&=\nonumber \frac{1}{d} \sum_{i_1 \dots i_n} \tr(Z_{i_1\dots i_n}) \tr(\omega^{\rm spec}_{i_1\dots i_n}) =\nonumber \frac{1}{d} \tr(\Sigma_{AB}) \geq 0.
\end{align}

Finally, we show that from $\mathcal{Z}$ with $C<0$ and 
$\tr(\Sigma_{AB})=0$ it is always possible to find a different 
solution $\mathcal{\bar Z}$ with $\bar C<0$ but $\tr(\Sigma_{AB})>0$
such that we can employ the rescaling trick again. Note first that
the only negative part in the $C$ must be due to 
$\tr(Z_{m\dots m} \omega_{m\dots m}^{\rm spec}) < 0$, since 
all other terms involve 
only positive semidefinite operators. Now pick any
$\omega_{i_1\dots i_n}^{\rm spec}$ with 
$\tr(\omega_{i_1\dots i_n}^{\rm spec}) > 0$, 
and assume this is $\omega^{\rm spec}_{am\dots m}$ with $a<m$. 
Then define the new set of operator 
\begin{align}
\bar Z_{am \dots m} &= Z_{am\dotsm} + \epsilon \mathbbm{1}, \nonumber \\
\bar Z_{mam \dots m} &= Z_{mam \dots m}, \dots, \bar Z_{m\dots m} = Z_{m\dots m}
\end{align}
which by Eq.~\eqref{eq:completeZ} are enough to fully determine 
the set $\mathcal{\bar Z}$. This set still contains only positive
semidefinite operators because the only operators that change are 
$\bar Z_{ai_2\dots i_n} = Z_{a i_2 \dots i_n} + \epsilon \mathbbm{1}$.
For this new solution $\mathcal{\bar Z}$ we get
$\tr(\bar \Sigma_{AB}) = \epsilon \tr(\omega^{\rm spec}_{am\dots m})$ 
and $\bar C = C+ \epsilon \tr(\omega^{\rm spec}_{am\dotsm})$, thus 
choosing $\epsilon$ small enough one obtains the given statement.
This completes the proof.
\end{proof}

\subsection{Proof of Proposition~2}

The ideas and bounding techniques are the same as in 
Ref.~[35], which derived similar determinant 
constraints for the dimension-bounded entanglement verification; 
here we only need to apply them to a single side. 

\begin{proof}
Inequality~(14) in the main text is just a
rearrangement of Eq.~(11) in the main text. We remark once more that the
bound of $T$ as given by Eq.~(10) in the main text holds only
if $\{ B_y \}_y$ is linearly independent, which follows
from the observation $|\det(D)| \not = 0$. 

The first and stronger condition in Eq.~(13) in the main text follows using 
the extra information of $C$ that if both sets $\{ G_k^A \}_k$, $\{ G_l^B \}_l$ 
have the identity in its linear span, then the largest singular value satisfies 
$\sigma_0(C) \geq q=\tr(\mathbbm{1}/\sqrt{d_A} \otimes \mathbbm{1}/\sqrt{d_B} \Sigma_{AB}) 
= 1/\sqrt{d_A d_B}$. This follows from the fact that the ordered singular values of 
$C$ are lower bounded by the ordered singular values of any submatrix $C^{\rm sub}$ 
of $C$. While $\{ G_l^B \}_l$ satisfies this extra condition automatically since 
$B_0=\mathbbm{1}$, we need this requirement for the choice of $\{ G_k^A \}_k$. 

Via this extra condition we can achieve a better bound using the inequality of 
arithmetic and geometric means only to $n_B$ singular values and then checking 
whether the minimal value of $\sigma_0(C)$ can be reached, more precisely one 
obtains
\begin{align}
\nonumber
\min_{\sigma_0(C) \geq q} \|C \|_1 \geq \min_{\sigma_0(C) \geq q} \Big[\sigma_0(C) 
+ \left(\frac{|\det(C)|}{\sigma_0(C)}\right)^{\frac{1}{n_B}}\Big] \\
\label{eq:help1}
= \left\{\! 
\begin{array}{cc}
(n_B + 1) |\det(C)|^{\frac{1}{n_B+1}} & \text{if } |\det(C)|^{\frac{1}{n_B+1}} \geq q \\
q + n_B \left(\frac{|\det(C)|}{q}\right)^{\frac{1}{n_B}} & \text{else} \end{array}, \right.
\end{align}
depending on the determinant of $C$. Note that both bounds are monotonically increasing 
functions. By the determinant rule $|\det(C)|=|\det(D)| | \det(T)|$ and the bound of 
Eq.~10 in the main text, the possible values are constrained to satisfy
\begin{equation}
\label{eq:region}
|\det(C)| \geq |\det(D)|d_B^{-\frac{n_B+1}{2}}.
\end{equation} 
Thus, depending on the value of $|\det(D)|$ the second bound in Eq.~\eqref{eq:help1}
can be used or not. If $|\det(D)|^{1/(n_B+1)} \geq 1/\sqrt{d_A}$ the determinant 
of $C$ will always satisfy the constraint in Eq.~\eqref{eq:help1} and one obtains
\begin{equation}
\label{eq:bound_region1}
\min_{\sigma_0(C) \geq q} \|C\|_1  \geq \frac{n_B+1}{\sqrt{d_B}} | \det(D) |^{\frac{1}{n_B+1}}.
\end{equation}
Otherwise one can split the possible region and minimize separately, yielding
\begin{align}
\label{eq:bound_region2}
 &\min_{\sigma_0(C) \geq q}\! \|C \|_1 \geq \\ \nonumber 
 &\min \!\left \{  \frac{1}{\sqrt{d_A d_B}}\! + n_B\!\left(\!\sqrt{d_A} d_B^{-\frac{n_B}{2}} | \det (D)|\right)^{\frac{1}{n_B}}\!\!, \frac{n_B+1}{\sqrt{d_A d_B}} \right\}\!. 
\end{align}

At last, if $n_B > \sqrt{d_A d_B} -1$ note that the bound given by 
Eq.~\eqref{eq:bound_region1} and the second argument in minimum of 
Eq.~\eqref{eq:bound_region2} are strictly larger than $1$. Thus only 
the first argument of Eq.~\eqref{eq:bound_region2} must be checked,
which is the stated condition. This completes the proof. \end{proof}

\subsection{Steering scenario for $n=2$ and $m=d$}

In this section we exemplify the construction of respective 
$\mathcal{Z}=\{ Z_{ij}\}_{ij}$ for the case of two settings
but arbitrary number of outcomes. The idea and construction 
rely on Fourier connected mutually unbiased bases~[33]. 
Thus we need a couple of definitions first.

Consider a Hilbert space $\mathbbm{C}^d$ and suppose that one 
has a basis $\{\ket{\phi_k}\}_{k\in \mathbbm{Z}_d}$ with 
$\mathbbm{Z}_d=\{ 0,\dots,d-1\}$, which we also use to label
the outcomes. Then one obtains another basis, which is mutually
unbiased, by the Fourier transform
\begin{equation}
\ket{\psi_k}=\mathcal{F}\ket{\phi_k}=\frac{1}{\sqrt{d}}\sum_{l \in \mathbbm{Z}_d} q^{kl} \ket{\phi_l}
\end{equation}  
with $q=e^{2\pi i/d}$.

These two bases even admit further structure which becomes convenient
in the following. Consider two representations $U,V$ of the cyclic
group $\mathbbm{Z}_d$ on $\mathcal{H}$ defined by its action onto 
the first basis, $U_x\ket{\psi_k}=\ket{\psi_{k+x}}$ and
$V_y \ket{\psi_k}=q^{yk}\ket{\psi_{k}}$ for all $x,y,k$. 
These two representations further satisfy $U_xV_y=q^{-xy}V_yU_x$
and the Fourier transform is the intertwining map, 
$U_x \mathcal{F}=\mathcal{F}V_x^\dag$ and $V_y \mathcal{F}=\mathcal{F} U_y$. 
Via this one can identify the action on both basis states that we summarize as
\begin{align}
\label{eq:rules1}
U_x \ket{\phi_k} &= \ket{\phi_{k+x}}, & U_x \ket{\psi_k}&=q^{-xk}\ket{\psi_k}, \\
\label{eq:rules2}
V_y \ket{\phi_k} &= q^{yk}\ket{\phi_{k}}, &V_y \ket{\psi_k}&=\ket{\psi_{k+y}}
\end{align}
for all $x,y \in \mathbbm{Z}_d$. Then the following set of operators
will be our characterization of the steering inequality. The structure 
can be guessed once one knows the so-called mother observable for the 
respective joint measurability problem~[33], from whose 
result one further knows that the current form is optimal.

\begin{proposition}
Consider the set of operators $\mathcal{Z}=\{Z_{kl}=U_kV_lZ_{00}V_l^\dag U_k^\dag\}$ 
with 
\begin{equation}
\label{eq:MUB_Z00}
Z_{00} = \mu_1\ket{\chi_-}\bra{\chi_-} + \mu_2(\mathbbm{1}-\ket{ \chi_+}\bra{\chi_+} - \ket{\chi_-}\bra{\chi_-}),
\end{equation}
pure states $\ket{\chi_{\pm}}\propto \ket{\phi_0}\pm\ket{\psi_0}$ and parameters
\begin{align}
\mu_1&=\frac{2}{\sqrt{d}(\sqrt{d}-1)(\sqrt{d}+2)}, \\ \mu_2&=\frac{1+\sqrt{d}}{\sqrt{d}(\sqrt{d}-1)(\sqrt{d}+2)}.
\end{align}
Then this set of operators can be used in the steering map, since all operators 
are positive semidefinite and uniquely determines the operator $\Sigma_{AB}$ and 
satisfies $\tr(\Sigma_{AB})=1$.
\end{proposition}

\begin{proof}
Using the form of $Z_{00}$ as given by Eq.~\eqref{eq:MUB_Z00} one sees 
that $Z_{00}$ is positive semidefinite, since both $\mu_i$ are strictly 
positive and $\ket{\chi_-}$ and $\ket{\chi_-}$ are orthogonal, moreover
it has unit trace. Since all other $Z_{kl}$ are obtained by a unitary 
transformation each $Z_{kl}$
is positive semidefinite and satisfies $\tr(Z_{kl})=1$, which directly shows 
that $\Sigma_{AB}$ has unit trace. Thus we are left to show that $Z_{kl}$ 
uniquely determines $\Sigma_{AB}$, for which we have to show
\begin{equation}
\label{eq:MUB_linrel}
Z_{kl}=Z_{kt}+Z_{sl}-Z_{st}
\end{equation}
for all $k,l,s,t \in \mathbbm{Z}_d$ according to Prop.~\ref{prop:Z}. In order 
to show this we expand the states $\ket{\chi_{\pm}}$ in $Z_{00}$ which results 
into the structure
\begin{equation}
Z_{00}=c_1 ( \ket{\phi_0}\bra{\phi_0}+\ket{\psi_0}\bra{\psi_0} ) + c_2 \mathbbm{1}
\end{equation}
with appropriate coefficients $c_1,c_2$. Note that at this point the very 
specific choices of $\mu_1$ and $\mu_2$ become important; they are chosen 
such that cross terms of  $\ket{\phi_0}\bra{\psi_0}$ or $\ket{\psi_0}\bra{\phi_0}$ 
vanish. Applying now the rules given by Eqs.~(\ref{eq:rules1}, \ref{eq:rules2}) one gets 
\begin{equation}
Z_{kl}=c_1(\ket{\phi_k}\bra{\phi_k}+\ket{\psi_l}\bra{\psi_l})+c_2 \mathbbm{1}
\end{equation}
from which the necessary relation given by Eq.~\eqref{eq:MUB_linrel} can be verified.
\end{proof}

In order to obtain a steering criterion one can use the given operators 
$Z_{kl}$ of the proposition to build up $\Sigma_{AB}$, which is uniquely 
determined by the given ensemble $\mathcal{E}$ in the $n=2$ and $m=d$ steering
case. Whenever this operator $\Sigma_{AB}$ is then not a separable state the 
underlying distribution is steerable.

\subsection{Dimension-bounded steering in a loophole free experiment of
Ref.~[19]}
First let us reiterate how to arrive at the data matrix necessary for employing the dimension bounded steering criterion. 
Alice and Bob have three different dichotomic measurements, 
$n_A=n_B=3$ and $m_A=m_B=2$, and we assume that Bob's measurement act 
onto a qubit $d_B=2$. The settings will be labeled by $x,y \in \{1,2,3\}$ 
and the outcomes by $a,b \in \{\pm 1\}$. 

According to Prop.2, let us first pick operators 
$Z_{ijk}$ with $i,j,k \in \{\pm 1\}$ that characterize a steering map 
with parameters $n_A=3$ and $m_A=2$. Here we choose 
%\begin{equation}
$Z_{ijk}=[\mathbbm{1} + \left( i \sigma_1 + j \sigma_2 + k \sigma_3)/\sqrt{3} \right]/2,$ 
%\end{equation}
which can be interpreted as pure states, whose Bloch vectors point towards 
the $8$ different corners of the cube. It can be checked that these choices 
satisfy all relations given by Eq.(8) of the main text, so that, by construction, 
the operator $\Sigma_{AB}$ is uniquely determined by the ensemble $\mathcal{E}$ 
and furthermore normalized. This operator is given by
\begin{equation}
\label{eq:state1}
\Sigma_{AB} = 
\frac{1}{2}\left[ \mathbbm{1} \otimes \rho + \frac{1}{\sqrt{3}} \sum_{s=1}^3 \sigma_s \otimes (\rho_{+|s}-\rho_{-|s}) \right].
\end{equation}

In order to get to the data matrix $D$ 
%$[ D(\Sigma_{AB})]_{ky}=\tr(G_k^A \otimes B_y \Sigma_{AB})$ 
we still 
need to fix the operator set $\{G_k^A\}_k$, for which the properly 
normalized identity and Pauli-operators, 
$\{ \mathbbm{1},\sigma_1,\sigma_2, \sigma_3\}/\sqrt{2}$, are convenient choices
since they only act non-trivially on certain terms in Eq.~\eqref{eq:state1}. Since only the subspace of $\{G_k^A\}_k$ matters in the criteria of Prop.2, any other basis choice will perform equally well. As the final step we rewrite the abstract values $\tr(B_y \rho_{a|x})$, 
with $B_0=\mathbbm{1}$ and $B_y=M_{+|y}-M_{-|y}$, in terms of the directly 
observable quantities $P(a,b|x,y)$. Looking at 
\begin{align*}
\tr[B_y &(\rho_{+|x}- \rho_{-|x})] \nonumber\\
=& \tr[(M_{+|y}-M_{-|y}) \rho_{+|x}] -  \tr[(M_{+|y}-M_{-|y}) \rho_{-|x}] \nonumber\\
=& P(+,+|x,y) - P(+,-|x,y) - \nonumber\\&\,[P(-,+|x,y) - P(-,-|x,y)] \equiv \braket{A_x B_y},
\end{align*}
one sees that correlations $\braket{A_xB_y}$ and respective marginals 
$\braket{A_x}, \braket{B_y}$, which similarly appear in Bell inequalities,
give an appropriate formulation. Hence, to sum up one gets the data matrix $D$
\begin{equation}
\nonumber
\small
\frac{1}{\sqrt{2}}\! \left[\! \begin{array}{cccc}
1 & \braket{B_1} & \braket{B_2} & \braket{B_3} \\
\braket{A_1}/\sqrt{3} & \braket{A_1 B_1}/\sqrt{3} & \braket{A_1 B_2}/\sqrt{3} & \braket{A_1 B_3}/\sqrt{3} \\
\braket{A_2}/\sqrt{3} & \braket{A_2 B_1}/\sqrt{3} & \braket{A_2 B_2}/\sqrt{3} & \braket{A_2 B_3}/\sqrt{3} \\
\braket{A_3}/\sqrt{3} & \braket{A_3 B_1}/\sqrt{3} & \braket{A_3 B_2}/\sqrt{3} & \braket{A_3 B_3}/\sqrt{3} \end{array} \!\right]. 
\end{equation}

Next let us explain how the developed criterion can be employed for the 
real setup used in Vienna~[19]. The main difference is that 
in the actual experiment one additionally observes an inconclusive outcome
``inc'' due to no click or even double click events. On Bob's side, the side
which is at least partially trusted, this event can safely be discarded~[19] 
assuming that this event is independent of the measurement choice such that 
it can be viewed as a kind of filter telling whether the final result will be 
conclusive or not. Only if this filter succeeds one looks at the corresponding 
state. For those measurements (acting on the conditional state) the measurements 
are assumed to act on a qubit, respective single photon in two polarization 
modes. 
However for Alice, the uncharacterized side, this is not possible. In order 
to incorporate the inconclusive event for Alice we consider the case that 
each inconclusive outcome ``inc'' is randomly assigned to either of the 
$+1$ or $-1$ outcome. This is also the standard for Bell experiments.
Then one is left with the dimension-bounded steering scenario considered
in the main section.

To finally give an example of the strength of our developed criterion we 
employ the following model to simulate real data: For the quantum state 
we assume a noisy maximally entangled singlet which has passed through 
a lossy channel for Alice, more precisely the state given by
\begin{align}
\nonumber
\rho_{AB} =& p\left[ \lambda \ket{\psi^-}\bra{\psi^-} + (1-\lambda) \mathbbm{1}/4 \right] \\&+ (1-p)\ket{\Omega}\bra{\Omega} \otimes \mathbbm{1}/2.
\end{align}
Here $p$ denotes the transmission probability, $\ket{\Omega}$ is the 
vacuum state and $\lambda$ a parameter characterizing the quality of 
the Werner state. In the true experiment there will be also loss on 
Bob's side, but as mentioned before, we look at the conditional state. 
Next we imagine that Alice and Bob perform projective measurements in 
the $\sigma_1,\sigma_2,\sigma_3$ basis, while the additional
``inc'' event for Alice is given by the projection onto the vacuum state.
Then the observed data, if Alice and Bob are using the same settings
$x,y$, are given by
\begin{align}
P(+,-|x,y)&=P(-,+|x,y)=\frac{1}{4}p(1+\lambda\delta_{x,y}), \\
P(+,+|x,y)&=P(-,-|x,y)=\frac{1}{4}p(1-\lambda\delta_{x,y}), \\
P(inc,+|x,y)&=P(inc,-|x,y)=\frac{1}{2}(1-p).
\end{align}
If one reassign each ``inc'' one obtains
\begin{align}
P(+,-|x,y)&=P(-,+|x,y)=\frac{1}{4}(1+p\lambda\delta_{x,y}), \\
P(+,+|x,y)&=P(-,-|x,y)=\frac{1}{4}(1-p\lambda\delta_{x,y}),
\end{align}
and thus
\begin{equation}
\braket{A_x B_y} = - \delta_{x,y} p \lambda, \:\:\braket{A_x} = \braket{B_y}= 0.
\end{equation}
Putting these observations into the data matrix from the main text 
one obtains
\begin{equation}
D = \left[ \begin{array}{cccc}
\frac{1}{\sqrt{2}} & 0 & 0 & 0 \\
0 & -\frac{p\lambda}{\sqrt{6}} & 0 & 0 \\
0 & 0 & -\frac{p\lambda}{\sqrt{6}} & 0  \\
0 & 0 & 0 & -\frac{p\lambda}{\sqrt{6}} \end{array} \right],
\end{equation}
which shows steering according to Eq.~(16) in the main text if $p\lambda > 1/\sqrt{3}\approx 0.577$.
Let us point out that this is also the condition if we would know that the performed 
measurements are perfect projective measurements in the eigenbasis of 
$\sigma_1,\sigma_2,\sigma_3$. Thus, we see that we have here a scenario 
where this further characterization is totally redundant and only the 
knowledge that one measures a qubit is essential.

Assuming the visibility and detection efficiency parameters from 
Ref.~[19], one would obtain the values $\{0.74,0.73,0.73\}$ 
for the respective $p\lambda$, which are all well above the threshold. 
Assuming that all other correlations and marginals vanish, this would 
strongly show steering also in the case where one has only the very 
limited knowledge that the conclusive outcomes were qubit measurements.
However, note, that these other observations are essential for the 
inequality, otherwise one could not gain the required extra knowledge 
of the uncharacterized qubit measurements. Unfortunately, these 
experimental data are not available anymore for the experiment of 
Ref.~[19].

\bibliographystyle{apsrev}

\begin{thebibliography}{39}
\expandafter\ifx\csname natexlab\endcsname\relax\def\natexlab#1{#1}\fi
\expandafter\ifx\csname bibnamefont\endcsname\relax
  \def\bibnamefont#1{#1}\fi
\expandafter\ifx\csname bibfnamefont\endcsname\relax
  \def\bibfnamefont#1{#1}\fi
\expandafter\ifx\csname citenamefont\endcsname\relax
  \def\citenamefont#1{#1}\fi
\expandafter\ifx\csname url\endcsname\relax
  \def\url#1{\texttt{#1}}\fi
\expandafter\ifx\csname urlprefix\endcsname\relax\def\urlprefix{URL }\fi
\providecommand{\bibinfo}[2]{#2}
\providecommand{\eprint}[2][]{\url{#2}}

\bibitem[{\citenamefont{Schr\"odinger}(1935)}]{schroedinger35a}
\bibinfo{author}{\bibfnamefont{E.}~\bibnamefont{Schr\"odinger}},
  \bibinfo{journal}{Proc. Camb. Phil. Soc.} \textbf{\bibinfo{volume}{31}},
  \bibinfo{pages}{555} (\bibinfo{year}{1935}).

\bibitem[{\citenamefont{Wiseman et~al.}(2007)\citenamefont{Wiseman, Jones, and
  Doherty}}]{wiseman07a}
\bibinfo{author}{\bibfnamefont{H.~M.} \bibnamefont{Wiseman}},
  \bibinfo{author}{\bibfnamefont{S.~J.} \bibnamefont{Jones}}, \bibnamefont{and}
  \bibinfo{author}{\bibfnamefont{A.~C.} \bibnamefont{Doherty}},
  \bibinfo{journal}{Phys. Rev. Lett.} \textbf{\bibinfo{volume}{98}},
  \bibinfo{pages}{140402} (\bibinfo{year}{2007}).

\bibitem[{\citenamefont{Brunner et~al.}(2014)\citenamefont{Brunner, Cavalcanti,
  Pironio, Scarani, and Wehner}}]{brunner_review}
\bibinfo{author}{\bibfnamefont{N.}~\bibnamefont{Brunner}},
  \bibinfo{author}{\bibfnamefont{D.}~\bibnamefont{Cavalcanti}},
  \bibinfo{author}{\bibfnamefont{S.}~\bibnamefont{Pironio}},
  \bibinfo{author}{\bibfnamefont{V.}~\bibnamefont{Scarani}}, \bibnamefont{and}
  \bibinfo{author}{\bibfnamefont{S.}~\bibnamefont{Wehner}},
  \bibinfo{journal}{Rev. Mod. Phys.} \textbf{\bibinfo{volume}{86}},
  \bibinfo{pages}{419} (\bibinfo{year}{2014}).

\bibitem[{\citenamefont{Bowles et~al.}(2014)\citenamefont{Bowles, V\'ertesi,
  Marco Túlio~Quintino, and Brunner}}]{bowles14a}
\bibinfo{author}{\bibfnamefont{J.}~\bibnamefont{Bowles}},
  \bibinfo{author}{\bibfnamefont{T.}~\bibnamefont{V\'ertesi}},
  \bibinfo{author}{\bibfnamefont{M.~T.} \bibnamefont{Quintino}},
  \bibnamefont{and} \bibinfo{author}{\bibfnamefont{N.}~\bibnamefont{Brunner}},
  \bibinfo{journal}{Phys. Rev. Lett.} \textbf{\bibinfo{volume}{112}},
  \bibinfo{pages}{200402} (\bibinfo{year}{2014}).

\bibitem[{\citenamefont{Quintino et~al.}(2014)\citenamefont{Quintino,
  V\'ertesi, and Brunner}}]{quintino14a}
\bibinfo{author}{\bibfnamefont{M.~T.} \bibnamefont{Quintino}},
  \bibinfo{author}{\bibfnamefont{T.}~\bibnamefont{V\'ertesi}},
  \bibnamefont{and} \bibinfo{author}{\bibfnamefont{N.}~\bibnamefont{Brunner}},
  \bibinfo{journal}{Phys. Rev. Lett.} \textbf{\bibinfo{volume}{113}},
  \bibinfo{pages}{160402} (\bibinfo{year}{2014}).

\bibitem[{\citenamefont{Uola et~al.}(2014)\citenamefont{Uola, Moroder, and
  G\"uhne}}]{uola14a}
\bibinfo{author}{\bibfnamefont{R.}~\bibnamefont{Uola}},
  \bibinfo{author}{\bibfnamefont{T.}~\bibnamefont{Moroder}}, \bibnamefont{and}
  \bibinfo{author}{\bibfnamefont{O.}~\bibnamefont{G\"uhne}},
  \bibinfo{journal}{Phys. Rev. Lett.} \textbf{\bibinfo{volume}{113}},
  \bibinfo{pages}{160403} (\bibinfo{year}{2014}).

\bibitem{piani14a}
\bibinfo{author}{\bibfnamefont{M.}~\bibnamefont{Piani}} \bibnamefont{and}
  \bibinfo{author}{\bibfnamefont{J.}~\bibnamefont{Watrous}},
  Phys. Rev. Lett. {\bf 114}, 060404 (2015).

\bibitem[{\citenamefont{Cavalcanti et~al.}(2009)\citenamefont{Cavalcanti,
  Jones, Wiseman, and Reid}}]{cavalcanti09a}
\bibinfo{author}{\bibfnamefont{E.~G.} \bibnamefont{Cavalcanti}},
  \bibinfo{author}{\bibfnamefont{S.~J.} \bibnamefont{Jones}},
  \bibinfo{author}{\bibfnamefont{H.~M.} \bibnamefont{Wiseman}},
  \bibnamefont{and} \bibinfo{author}{\bibfnamefont{M.~D.} \bibnamefont{Reid}},
  \bibinfo{journal}{Phys. Rev. A} \textbf{\bibinfo{volume}{80}},
  \bibinfo{pages}{032112} (\bibinfo{year}{2009}).

\bibitem[{\citenamefont{Schneeloch et~al.}(2013)\citenamefont{Schneeloch,
  Broadbent, Walborn, Cavalcanti, and Howell}}]{schneeloch13a}
\bibinfo{author}{\bibfnamefont{J.}~\bibnamefont{Schneeloch}},
  \bibinfo{author}{\bibfnamefont{C.~J.} \bibnamefont{Broadbent}},
  \bibinfo{author}{\bibfnamefont{S.~P.} \bibnamefont{Walborn}},
  \bibinfo{author}{\bibfnamefont{E.~G.} \bibnamefont{Cavalcanti}},
  \bibnamefont{and} \bibinfo{author}{\bibfnamefont{J.~C.}
  \bibnamefont{Howell}}, \bibinfo{journal}{Phys. Rev. A}
  \textbf{\bibinfo{volume}{87}}, \bibinfo{pages}{062103}
  (\bibinfo{year}{2013}).

\bibitem[{\citenamefont{Evans et~al.}(2013)\citenamefont{Evans, Cavalcanti, and
  Wiseman}}]{evans13a}
\bibinfo{author}{\bibfnamefont{D.~A.} \bibnamefont{Evans}},
  \bibinfo{author}{\bibfnamefont{E.~G.} \bibnamefont{Cavalcanti}},
  \bibnamefont{and} \bibinfo{author}{\bibfnamefont{H.~M.}
  \bibnamefont{Wiseman}}, \bibinfo{journal}{Phys. Rev. A}
  \textbf{\bibinfo{volume}{88}}, \bibinfo{pages}{022106}
  (\bibinfo{year}{2013}).

\bibitem{horodecki14a}
M. Horodecki, M. Marciniak, and Z. Yin,
J. Phys. A: Math. Theor. {\bf 48}, 135303 (2015).

\bibitem[{\citenamefont{Zukowski et~al.}()\citenamefont{Zukowski, Dutta, and
  Yin}}]{zukowski14a}
\bibinfo{author}{\bibfnamefont{M.}~\bibnamefont{Zukowski}},
  \bibinfo{author}{\bibfnamefont{A.}~\bibnamefont{Dutta}}, \bibnamefont{and}
  \bibinfo{author}{\bibfnamefont{Z.}~\bibnamefont{Yin}},
  Phys. Rev. A {\bf 91}, 032107 (2015).

\bibitem{marciniak14a}
M. Marciniak, A. Rutkowski, Z. Yin, M. Horodecki, 
and R. Horodecki,
Phys. Rev. Lett. {\bf 115}, 170401 (2015).

\bibitem{skrzypczyk13a}
P. Skrzypczyk, M. Navascues, and D. Cavalcanti,
Phys. Rev. Lett. {\bf 112}, 180404 (2014).

\bibitem{steeringGauss}
I. Kogias and G. Adesso,
J. Opt. Soc. Am. B {\bf 32}, A27 (2015)

\bibitem[{\citenamefont{Pusey}(2013)}]{pusey13a}
\bibinfo{author}{\bibfnamefont{M.~A.} \bibnamefont{Pusey}},
  \bibinfo{journal}{Phys. Rev. A} \textbf{\bibinfo{volume}{88}},
  \bibinfo{pages}{032313} (\bibinfo{year}{2013}).

\bibitem[{\citenamefont{Moroder et~al.}(2014)\citenamefont{Moroder, Gittsovich,
  Huber, and G\"uhne}}]{moroder14a}
\bibinfo{author}{\bibfnamefont{T.}~\bibnamefont{Moroder}},
  \bibinfo{author}{\bibfnamefont{O.}~\bibnamefont{Gittsovich}},
  \bibinfo{author}{\bibfnamefont{M.}~\bibnamefont{Huber}}, \bibnamefont{and}
  \bibinfo{author}{\bibfnamefont{O.}~\bibnamefont{G\"uhne}},
  \bibinfo{journal}{Phys. Rev. Lett.} \textbf{\bibinfo{volume}{113}},
  \bibinfo{pages}{050404} (\bibinfo{year}{2014}).

\bibitem[{\citenamefont{Bennet et~al.}(2012)\citenamefont{Bennet, Evans,
  Saunders, Branciard, Cavalcanti, Wiseman, and Pryde}}]{bennet12a}
\bibinfo{author}{\bibfnamefont{A.~J.} \bibnamefont{Bennet}},
  \bibinfo{author}{\bibfnamefont{D.~A.} \bibnamefont{Evans}},
  \bibinfo{author}{\bibfnamefont{D.~J.} \bibnamefont{Saunders}},
  \bibinfo{author}{\bibfnamefont{C.}~\bibnamefont{Branciard}},
  \bibinfo{author}{\bibfnamefont{E.~G.} \bibnamefont{Cavalcanti}},
  \bibinfo{author}{\bibfnamefont{H.~M.} \bibnamefont{Wiseman}},
  \bibnamefont{and} \bibinfo{author}{\bibfnamefont{G.~J.} \bibnamefont{Pryde}},
  \bibinfo{journal}{Phys. Rev. X} \textbf{\bibinfo{volume}{2}},
  \bibinfo{pages}{031003} (\bibinfo{year}{2012}).

\bibitem[{\citenamefont{Wittmann et~al.}(2012)\citenamefont{Wittmann, Ramelow,
  Steinlechner, Langford, Brunner, Wiseman, Ursin, and
  Zeilinger}}]{wittmann12a}
\bibinfo{author}{\bibfnamefont{B.}~\bibnamefont{Wittmann}},
  \bibinfo{author}{\bibfnamefont{S.}~\bibnamefont{Ramelow}},
  \bibinfo{author}{\bibfnamefont{F.}~\bibnamefont{Steinlechner}},
  \bibinfo{author}{\bibfnamefont{N.~K.} \bibnamefont{Langford}},
  \bibinfo{author}{\bibfnamefont{N.}~\bibnamefont{Brunner}},
  \bibinfo{author}{\bibfnamefont{H.}~\bibnamefont{Wiseman}},
  \bibinfo{author}{\bibfnamefont{R.}~\bibnamefont{Ursin}}, \bibnamefont{and}
  \bibinfo{author}{\bibfnamefont{A.}~\bibnamefont{Zeilinger}},
  \bibinfo{journal}{New J. Phys.} \textbf{\bibinfo{volume}{14}},
  \bibinfo{pages}{053030} (\bibinfo{year}{2012}).

\bibitem[{\citenamefont{Smith et~al.}(2012)\citenamefont{Smith, Gillett,
  de~Almeida, Branciard, Fedrizzi, Weinhold, Lita, Calkins, Gerrits, Wiseman
  et~al.}}]{smith12a}
\bibinfo{author}{\bibfnamefont{D.~H.} \bibnamefont{Smith}},
  \bibinfo{author}{\bibfnamefont{G.}~\bibnamefont{Gillett}},
  \bibinfo{author}{\bibfnamefont{M.}~\bibnamefont{de~Almeida}},
  \bibinfo{author}{\bibfnamefont{C.}~\bibnamefont{Branciard}},
  \bibinfo{author}{\bibfnamefont{A.}~\bibnamefont{Fedrizzi}},
  \bibinfo{author}{\bibfnamefont{T.~J.} \bibnamefont{Weinhold}},
  \bibinfo{author}{\bibfnamefont{A.}~\bibnamefont{Lita}},
  \bibinfo{author}{\bibfnamefont{B.}~\bibnamefont{Calkins}},
  \bibinfo{author}{\bibfnamefont{T.}~\bibnamefont{Gerrits}},
  \bibinfo{author}{\bibfnamefont{H.~M.} \bibnamefont{Wiseman}},
  \bibnamefont{et~al.}, \bibinfo{journal}{Nat. Comm.}
  \textbf{\bibinfo{volume}{3}}, \bibinfo{pages}{625} (\bibinfo{year}{2012}).

\bibitem[{\citenamefont{Horodecki et~al.}(2009)\citenamefont{Horodecki,
  Horodecki, Horodecki, and Horodecki}}]{horodecki_review}
\bibinfo{author}{\bibfnamefont{R.}~\bibnamefont{Horodecki}},
  \bibinfo{author}{\bibfnamefont{P.}~\bibnamefont{Horodecki}},
  \bibinfo{author}{\bibfnamefont{M.}~\bibnamefont{Horodecki}},
  \bibnamefont{and}
  \bibinfo{author}{\bibfnamefont{K.}~\bibnamefont{Horodecki}},
  \bibinfo{journal}{Rev. Mod. Phys.} \textbf{\bibinfo{volume}{81}},
  \bibinfo{pages}{2009} (\bibinfo{year}{2009}).

\bibitem[{\citenamefont{G\"uhne and T\'oth}(2009)}]{guehne09a}
\bibinfo{author}{\bibfnamefont{O.}~\bibnamefont{G\"uhne}} \bibnamefont{and}
  \bibinfo{author}{\bibfnamefont{G.}~\bibnamefont{T\'oth}},
  \bibinfo{journal}{Physics Reports} \textbf{\bibinfo{volume}{474}},
  \bibinfo{pages}{1} (\bibinfo{year}{2009}).

\bibitem[{\citenamefont{Piani et~al.}(2008)\citenamefont{Piani, Horodecki, and
  Horodecki}}]{piani08a}
\bibinfo{author}{\bibfnamefont{M.}~\bibnamefont{Piani}},
  \bibinfo{author}{\bibfnamefont{P.}~\bibnamefont{Horodecki}},
  \bibnamefont{and}
  \bibinfo{author}{\bibfnamefont{R.}~\bibnamefont{Horodecki}},
  \bibinfo{journal}{Phys. Rev. Lett.} \textbf{\bibinfo{volume}{100}},
  \bibinfo{pages}{090502} (\bibinfo{year}{2008}).

\bibitem[{\citenamefont{Dakic et~al.}(2010)\citenamefont{Dakic, Vedral, and
  Brukner}}]{dakic10a}
\bibinfo{author}{\bibfnamefont{B.}~\bibnamefont{Dakic}},
  \bibinfo{author}{\bibfnamefont{V.}~\bibnamefont{Vedral}}, \bibnamefont{and}
  \bibinfo{author}{\bibfnamefont{C.}~\bibnamefont{Brukner}},
  \bibinfo{journal}{Phys. Rev. Lett.} \textbf{\bibinfo{volume}{105}},
  \bibinfo{pages}{190502} (\bibinfo{year}{2010}).

\bibitem[{\citenamefont{Rat}()}]{navascuesrat}
\bibinfo{author}{\bibfnamefont{S.}~\bibnamefont{Rat}}, \bibinfo{note}{{\it
  {Research} lines that lead nowhere (I): Quantum Discord}, available at
  http://schroedingersrat. blogspot.com}.

  \bibitem[{\citenamefont{Vandenberghe and Boyd}(1996)}]{vandenberghe96a}
\bibinfo{author}{\bibfnamefont{L.}~\bibnamefont{Vandenberghe}}
  \bibnamefont{and} \bibinfo{author}{\bibfnamefont{S.}~\bibnamefont{Boyd}},
  \bibinfo{journal}{SIAM Review} \textbf{\bibinfo{volume}{38}},
  \bibinfo{pages}{49} (\bibinfo{year}{1996}).
  
\bibitem[{\citenamefont{Peres}(1996)}]{peres96a}
\bibinfo{author}{\bibfnamefont{A.}~\bibnamefont{Peres}},
  \bibinfo{journal}{Phys. Rev. Lett.} \textbf{\bibinfo{volume}{77}},
  \bibinfo{pages}{1413} (\bibinfo{year}{1996}).

\bibitem[{\citenamefont{Horodecki et~al.}(1996)\citenamefont{Horodecki,
  Horodecki, and Horodecki}}]{horodecki96b}
\bibinfo{author}{\bibfnamefont{M.}~\bibnamefont{Horodecki}},
  \bibinfo{author}{\bibfnamefont{P.}~\bibnamefont{Horodecki}},
  \bibnamefont{and}
  \bibinfo{author}{\bibfnamefont{R.}~\bibnamefont{Horodecki}},
  \bibinfo{journal}{Phys. Lett. A} \textbf{\bibinfo{volume}{223}},
  \bibinfo{pages}{1} (\bibinfo{year}{1996}).

\bibitem[{\citenamefont{Terhal}(2000)}]{terhal00a}
\bibinfo{author}{\bibfnamefont{B.}~\bibnamefont{Terhal}},
  \bibinfo{journal}{Phys. Lett. A} \textbf{\bibinfo{volume}{271}},
  \bibinfo{pages}{319} (\bibinfo{year}{2000}).

\bibitem[{\citenamefont{Rudolph}(2005)}]{rudolph02a}
\bibinfo{author}{\bibfnamefont{O.}~\bibnamefont{Rudolph}},
  \bibinfo{journal}{Quantum Inf. Proc.} \textbf{\bibinfo{volume}{4}},
  \bibinfo{pages}{219} (\bibinfo{year}{2005}).

\bibitem[{\citenamefont{Chen and Wu}(2003)}]{chen03a}
\bibinfo{author}{\bibfnamefont{K.}~\bibnamefont{Chen}} \bibnamefont{and}
  \bibinfo{author}{\bibfnamefont{L.-A.} \bibnamefont{Wu}},
  \bibinfo{journal}{Quant. Inf. Comp.} \textbf{\bibinfo{volume}{3}},
  \bibinfo{pages}{193} (\bibinfo{year}{2003}).

\bibitem[{\citenamefont{G\"uhne et~al.}(2007)\citenamefont{G\"uhne, Hyllus,
  Gittsovich, and Eisert}}]{PhysRevLett.99.130504}
\bibinfo{author}{\bibfnamefont{O.}~\bibnamefont{G\"uhne}},
  \bibinfo{author}{\bibfnamefont{P.}~\bibnamefont{Hyllus}},
  \bibinfo{author}{\bibfnamefont{O.}~\bibnamefont{Gittsovich}},
  \bibnamefont{and} \bibinfo{author}{\bibfnamefont{J.}~\bibnamefont{Eisert}},
  \bibinfo{journal}{Phys. Rev. Lett.} \textbf{\bibinfo{volume}{99}},
  \bibinfo{pages}{130504} (\bibinfo{year}{2007}).



\bibitem[{\citenamefont{Carmeli et~al.}(2012)\citenamefont{Carmeli, Heinosaari,
  and Toigo}}]{carmeli12a}
\bibinfo{author}{\bibfnamefont{C.}~\bibnamefont{Carmeli}},
  \bibinfo{author}{\bibfnamefont{T.}~\bibnamefont{Heinosaari}},
  \bibnamefont{and} \bibinfo{author}{\bibfnamefont{A.}~\bibnamefont{Toigo}},
  \bibinfo{journal}{Phys. Rev. A} \textbf{\bibinfo{volume}{85}},
  \bibinfo{pages}{012109} (\bibinfo{year}{2012}).

\bibitem[{\citenamefont{Brunner et~al.}(2008)\citenamefont{Brunner, Pironio,
  Ac{\'\i}n, Nicolas~Gisin, M\'ethot, and Scarani}}]{brunner08a}
\bibinfo{author}{\bibfnamefont{N.}~\bibnamefont{Brunner}},
  \bibinfo{author}{\bibfnamefont{S.}~\bibnamefont{Pironio}},
  \bibinfo{author}{\bibfnamefont{A.}~\bibnamefont{Ac{\'\i}n}},
  \bibinfo{author}{\bibfnamefont{N.}~\bibnamefont{Nicolas~Gisin}},
  \bibinfo{author}{\bibfnamefont{A.~A.} \bibnamefont{M\'ethot}},
  \bibnamefont{and} \bibinfo{author}{\bibfnamefont{V.}~\bibnamefont{Scarani}},
  \bibinfo{journal}{Phys. Rev. Lett.} \textbf{\bibinfo{volume}{100}},
  \bibinfo{pages}{210503} (\bibinfo{year}{2008}).

\bibitem[{\citenamefont{Moroder and Gittsovich}(2012)}]{moroder12a}
\bibinfo{author}{\bibfnamefont{T.}~\bibnamefont{Moroder}} \bibnamefont{and}
  \bibinfo{author}{\bibfnamefont{O.}~\bibnamefont{Gittsovich}},
  \bibinfo{journal}{Phys. Rev. A} \textbf{\bibinfo{volume}{85}},
  \bibinfo{pages}{032301} (\bibinfo{year}{2012}).

\bibitem[{\citenamefont{Gittsovich and Moroder}()}]{gittsovich12a}
\bibinfo{author}{\bibfnamefont{O.}~\bibnamefont{Gittsovich}} \bibnamefont{and}
  \bibinfo{author}{\bibfnamefont{T.}~\bibnamefont{Moroder}},
  \bibinfo{note}{arXiv:1303.3484, QCMC Proceedings}.

\bibitem[{\citenamefont{Woodhead and Pironio}(2013)}]{woodhead13a}
\bibinfo{author}{\bibfnamefont{E.}~\bibnamefont{Woodhead}} \bibnamefont{and}
  \bibinfo{author}{\bibfnamefont{S.}~\bibnamefont{Pironio}},
  \bibinfo{journal}{Phys. Rev. A} \textbf{\bibinfo{volume}{87}},
  \bibinfo{pages}{032315} (\bibinfo{year}{2013}).

\bibitem[{\citenamefont{de~Vicente}(2014)}]{julio}
\bibinfo{author}{\bibfnamefont{J.~I.} \bibnamefont{de~Vicente}},
  \bibinfo{journal}{J. Phys. A: Math. Theor.} \textbf{\bibinfo{volume}{47}},
  \bibinfo{pages}{424017} (\bibinfo{year}{2014}).

\bibitem{rodrigo}
R. Gallego and L. Aolita,
Phys. Rev. X {\bf 5}, 041008 (2015)

\end{thebibliography}

\end{document}